\documentclass{jpp}
\usepackage{graphicx}
\usepackage{amsmath,amsfonts,amssymb}
\usepackage{mathrsfs}
\usepackage{breqn}
\usepackage{bm}
\usepackage[numbers,square]{natbib}
\usepackage{hyperref}
\usepackage{cleveref}

\newtheorem{definition}{Definition}
\newtheorem{proposition}{Proposition}

\makeatletter
\newcommand*{\defeq}{\mathrel{\rlap{%
					 \raisebox{0.3ex}{$\m@th\cdot$}}%
					 \raisebox{-0.3ex}{$\m@th\cdot$}}%
					 =}
\makeatother

\newcommand{\ee}{\operatorname{e}}
\newcommand{\ii}{\operatorname{i}}
\newcommand{\dd}{\operatorname{d}}
\newcommand{\bj}{\mathbf{j}}
\newcommand{\bk}{\mathbf{k}}
\newcommand{\bp}{\mathbf{p}}
\newcommand{\bq}{\mathbf{q}}

\newcommand{\bv}{\mathbf{v}}
\newcommand{\bx}{\mathbf{x}}
\newcommand{\bE}{\mathbf{E}}
\newcommand{\bB}{\mathbf{B}}

\newcommand{\bb}{\mathbf{b}}

\newcommand{\vperp}{v_\perp}
\newcommand{\vpar}{v_\parallel}

\newcommand{\wc}{\omega}

\shorttitle{Bivariate kernels and nonlinear Vlasov--Maxwell response}
\shortauthor{R. Ricci}

\title{Bivariate incomplete-Bessel kernels for the first nonlinear Vlasov--Maxwell response}

\author{Roberto Ricci\aff{1}\corresp{\email{roberto.ricci@enea.it}}}

\affiliation{\aff{1}ENEA, Nuclear Department NUC-DTT, Frascati Research Center, Via E. Fermi 45, 00044 Frascati (Rome), Italy}

\begin{document}

\maketitle

\begin{abstract}
The weakly nonlinear response of a homogeneous magnetised plasma is usually written as a double cyclotron-harmonic expansion. This representation is explicit, but the repeated use of the Jacobi--Anger expansion produces long sums of Bessel functions and nested resonance denominators. In the linear problem we recently obtained an alternative formulation by keeping the Larmor phase unexpanded and by evaluating the characteristic integral in terms of the newly introduced incomplete-Bessel function \(G_\mu(z,\psi)\). In this article we apply the same idea directly to the first nonlinear Vlasov--Maxwell equation. The nonlinear source contains the linear response of an inner mode. When this linear response is written in incomplete-Bessel form, the outer characteristic integral produces a bivariate orbit-resolvent \(G_{\mu,\nu}^{(r)}(z,\psi;w,\chi)\). This function emerges therefore as the natural orbit integral generated by the first nonlinear characteristic problem. We derive the nonlinear distribution function in terms of these bivariate functions, collect the identities needed for the current projection, recover an extension of the classical double-harmonic Liu--Tripathi formula by expansion, and indicate how the nonlinear susceptibility tensor is obtained from bivariate angular contractions.

\end{abstract}


\section{Introduction}\label{sec:introduction}

The kinetic theory of nonlinear wave coupling in magnetised plasmas is a classical problem. The standard weakly nonlinear treatment expands the Vlasov--Maxwell system about a homogeneous magnetised equilibrium and expresses the response in cyclotron harmonics. This gives the familiar picture of resonant mode coupling, but the explicit formulas rapidly become cumbersome when electromagnetic waves propagate at arbitrary angle with respect to the background magnetic field. In the first nonlinear correction one already encounters two coupled harmonic indices: one from the outer characteristic integration of the output mode and one from the inner linear response entering the nonlinear source. This is the algebraic origin of the complicated double sums appearing in Liu--Tripathi and similar calculations \citep{LiuTripathi1986,Sugaya1989ResonantWaveWave,Yoon2024NonlinearSusceptibilities}.

In the linear problem, the same cyclotron structure can be organised in a different way. Instead of expanding the Larmor phase at the beginning of the calculation, one keeps the phase along the unperturbed orbit and evaluates the characteristic integral directly. This produces the incomplete-Bessel function \(G_\mu(z,\psi)\), which replaces the cyclotron harmonic sum before the susceptibility tensor is projected \citep{Ricci2026Bessel}. This idea is closely related to the sum-rule approach of Newberger and subsequent developments \citep{Newberger1982,BakkerTemme1984,Swanson2003,QinPhillipsDavidson2007,LercheSchlickeiserTautz2008,QinPhillipsDavidson2008Response}.

The purpose of the present paper is to show that the first nonlinear problem admits an exactly analogous organisation. The first nonlinear equation has the same left-hand side as the linear equation; only the source term changes. More precisely, the nonlinear source is the Lorentz force of one mode acting on the linear response of another mode. If the inner linear response is written in its incomplete-Bessel form, then the outer characteristic integral acts on a shifted \(G_{\mu_\bq}\)-response. The resulting object is a bivariate incomplete-Bessel function
\begin{dmath*}
	G_{\mu_\bk,\mu_\bq}^{(r)}
	\left(
		z_\bk,\psi_\bk;
		z_\bq,\psi_\bq
	\right).
\end{dmath*}
The integer \(r\) records the finite harmonic shifts generated by transverse velocity factors and gyrophase derivatives in the nonlinear force.

The terminology ``bivariate incomplete-Bessel function'' is used in the structural sense: $G_{\mu,\nu}^{(r)}$ is the bivariate orbit-resolvent generated by applying a second incomplete-Bessel-type characteristic transform to the one-variable incomplete-Bessel kernel $G_\nu$. It is not meant to imply that the function is a classical two-variable Bessel function in the sense of the standard theory of multivariable Bessel functions.

This changes the exposition of the nonlinear calculation. The bivariate function appears directly from the nonlinear characteristic solution, exactly as \(G_\mu\) appears directly from the linear characteristic solution. The double-harmonic formula is then recovered by expanding the bivariate function. In this sense the traditional harmonic expression becomes a validation of the orbit formulation rather than its starting point.

The article is organised as follows. We first state the perturbative Vlasov--Maxwell problem and introduce the notation for the cyclotron orbit. We recall the linear solution in incomplete-Bessel form, using the integral definition of \(G_\mu\) as primary. We then take the nonlinear characteristic formula from the standard calculation and rewrite its source with compact differential operators. Inserting the \(G_{\mu_\bq}\)-form of the linear response produces the bivariate incomplete-Bessel function. We derive the explicit nonlinear distribution function in this notation, study the main properties of the bivariate functions, derive independently the traditional double-harmonic formula and recover the same formula by expansion. We finally describe the current projection leading to the nonlinear susceptibility tensor, by providing an explicit compact formula.

\section{Perturbative Vlasov--Maxwell problem and orbit notation}\label{sec:problem}

We consider a homogeneous plasma in a uniform background magnetic field
\begin{dmath*}
	\bB_0
	=
	B_0\bb,
\end{dmath*}
where \(\bb\) is a fixed unit vector, assumed along the $z$ direction of an otherwise arbitrarily oriented Cartesian frame. For one species, the Vlasov equation is
\begin{dmath*}
	\frac{\partial F}{\partial t}
	+
	\bv\cdot\partial_\bx F
	+
	\frac{e}{m}
	\left(
		\bE+
		\bv\times\bB
	\right)\cdot\partial_\bv F
	=
	0
\end{dmath*}.
Here $F$ is the species distribution function, \(e\) is the signed charge and \(m\) denotes the particle mass. We write
\begin{dmath*}
	\bB
	= {
	\bB_0+
	\bB^{(1)}+
	\bB^{(2)}+
	\cdots,
	\qquad
	\bE
	=
	\bE^{(1)}+
	\bE^{(2)}+
	\cdots
	}
\end{dmath*}
and
\begin{dmath*}
	F
	=
	f_0+
	f^{(1)}+
	f^{(2)}+
	\cdots.
\end{dmath*}
The equilibrium distribution \(f_0=f_0(\vperp,\vpar)\) is assumed gyrotropic. 
Then the linear problem is described by
\begin{dmath*}
	\left[\partial_t+\bv\cdot\partial_{\bx}+\frac{e}{m}(\bv\times\bB_0)\cdot\partial_{\bv}\right]f^{(1)}(\bx,\bv,t)=-\frac{e}{m}\left[\bv\times\bB_0\right]\cdot\partial_{\bv}f_0(v_\perp, v_\parallel)
\end{dmath*}
The first non linear equation considered in this work has the same structure, apart from the source term:
\begin{dmath*}
	\left[\partial_t+\bv\cdot\partial_{\bx}+\frac{e}{m}(\bv\times\bB_0)\cdot\partial_{\bv}\right]f^{(2)}(\bx,\bv,t)=-\frac{e}{m}\left[\bE^{(1)}(\bx,t)+\bv\times\bB^{(1)}(\bx,t)\right]\cdot\partial_{\bv}f^{(1)}(\bx, \bv, t)
\end{dmath*}.
At all perturbative levels, the left-hand side is the same unperturbed cyclotron operator. 
After Laplace--Fourier transforming, the equations can be solved by using the method of characteristics.

We adopt the following conventions. The unperturbed signed cyclotron frequency is
\begin{dmath*}
	\wc
	\defeq
	\frac{eB_0}{m}.
\end{dmath*}
For a Laplace--Fourier mode \((\bk,\Omega)\), the characteristic parameters are
\begin{dmath*}
	z_\bk
	\defeq {
	\frac{k_\perp v_\perp}{\wc},
	\qquad
	\mu_{\bk, \Omega}
	\defeq
	\frac{k_\parallel v_\parallel-\Omega}{\wc},
	\qquad
	\psi_\bk
	\defeq
	\phi-\phi_\bk
	}
\end{dmath*}.
In the following, we set $\Omega = \Omega_\bk$, where $\Omega_\bk = \omega_\bk + \ii \gamma_\bk$ denotes a generally complex solution of the linear dispersion equation, assumed one-branched for simplicity.

Along the unperturbed orbit we use
\begin{dmath*}
	\lambda
	= {
	\wc\tau,
	\qquad
	\psi_\bk(\lambda)
	=
	\psi_\bk+
	\lambda
	}
\end{dmath*}.
The linear and nonlinear responses at mode \(\bk\) have the same external Larmor dressing,
\begin{dmath*}
	f_\bk^{(s)}(\bv)
	=
	\ee^{\ii z_\bk\sin\psi_\bk}
	\tilde f_\bk^{(s)}(\bv),
\end{dmath*}
where \(s=1,2,\ldots\). This factorisation is a consequence of the common cyclotron characteristic operator. It is the basic reason why the incomplete-Bessel construction can be iterated from the linear to the first nonlinear response.

We drop for simplicity the superscript from $\bE^{(1)}$ and use circular transverse components:
\begin{dmath*}
	E_{\bk,1}
	\defeq {
	\frac{E_{\bk,x}-\ii E_{\bk,y}}{\sqrt2},
	\qquad
	E_{\bk,-1}
	\defeq
	\frac{E_{\bk,x}+\ii E_{\bk,y}}{\sqrt2},
	\qquad 
	E_{\bk,0} = E_{\bk,z}
	}
\end{dmath*}.
The two velocity-space differential operators used below are
\begin{dmath*}
	\mathcal G_\bk
	\defeq
	\partial_{v_\perp}
	+
	\frac{k_\parallel}{\Omega_\bk}
	\left(
		v_\perp\partial_{v_\parallel}
		-
		v_\parallel\partial_{v_\perp}
	\right),
\end{dmath*}
and
\begin{dmath*}
	\mathcal G
	\defeq
	v_\parallel\partial_{v_\perp}
	-
	v_\perp\partial_{v_\parallel}
\end{dmath*}.

For the first nonlinear response we impose the resonance condition
\begin{dmath*}
	\bk
	= {
	\bp+
	\bq,
	\qquad
	\Omega_\bk
	=
	\Omega_\bp+
	\Omega_\bq
	}
\end{dmath*}.

The prime on sums over \(\bp\) and \(\bq\) denotes the standard exclusion of zero modes.

\section{The linear solution in incomplete-Bessel form}\label{sec:linear}

The one-variable incomplete-Bessel function introduced in \citep{Ricci2026Bessel} is primarily defined as the orbit integral
\begin{dmath*}
	G_\mu(z,\psi)
	\defeq
	\ii
	\int_0^\infty
	\ee^{-\ii z\sin(\psi+\lambda)}
	\ee^{-\ii\mu\lambda}
	\dd\lambda,
\end{dmath*}
with the usual causal prescription understood in the oscillatory integral. This function is $2\pi$-periodic in the variable $\psi$. Its harmonic representation is obtained by using the Jacobi--Anger formula \cite[Sec.~10.12, Eqs.~10.12.1--10.12.2]{NIST:DLMF} to expand the Larmor phase:
\begin{dmath*}
	G_\mu(z,\psi)
	=
	\sum_{n=-\infty}^{\infty}
	\frac{J_n(z)\ee^{-\ii n\psi}}{n+\mu}.
\end{dmath*}
This is the linear prototype: one cyclotron characteristic integral acting on one Larmor phase produces one cyclotron denominator.

The linear characteristic solution can be written as
\begin{dmath*}
	f_\bq^{(1)}(\bv)
	=
	\ee^{\ii z_\bq\sin\psi_\bq}
	\tilde f_\bq^{(1)}(\bv)
\end{dmath*},
where
\begin{dmath*}
	\tilde f_\bq^{(1)}
	=
	\frac{\ii e}{m\wc}
	\sum_{a \in \mathcal I}
		\ee^{\ii a\psi_\bq}
		G_{\mu_\bq-a}
		\left(
			z_\bq,
			\psi_\bq
		\right)
		X_{\bq,a}(v_\perp, v_\parallel)
\end{dmath*}.
Here $\mathcal I \defeq \{-1,0,1\}$ and
\begin{dmath*}
	X_{\bq,1}
	\defeq
	\left(
	\frac{\ee^{\ii\phi_\bq}}{\sqrt2}E_{\bq,1}\mathcal G_\bq
	+
	\frac{q_\perp E_{\bq,0}}{2\Omega_\bq}\mathcal G
	\right)
	f_0
\end{dmath*}
\begin{dmath*}
	X_{\bq,-1}
	\defeq
	\left(
	\frac{\ee^{-\ii\phi_\bq}}{\sqrt2}E_{\bq,-1}\mathcal G_\bq
	+
	\frac{q_\perp E_{\bq,0}}{2\Omega_\bq}\mathcal G
	\right)
	f_0
\end{dmath*}
and
\begin{dmath*}
	X_{\bq,0}
	\defeq
	E_{\bq,0}\partial_{v_\parallel} f_0.
\end{dmath*}
The shifted factors \(G_{\mu_\bq-1}\), \(G_{\mu_\bq+1}\), and \(G_{\mu_\bq}\) are the result of the linear orbit integral and multiply the corresponding velocity-space source terms.

The standard harmonic coefficients are recovered by expanding the shifted incomplete-Bessel factors. The familiar coefficient form is
\begin{dmath*}
	f_\bq^{(1)}(\bv)
	=
	\ee^{\ii z_\bq\sin\psi_\bq}
	\sum_{N=-\infty}^{\infty}
	\ee^{-\ii N\psi_\bq}
	\tilde f_{\bq,N}^{(1)}(v_\perp,v_\parallel),
\end{dmath*}
with
\begin{dmath*}
	\tilde f_{\bq,N}^{(1)}
	=
	\ii\frac{e}{m}
	\frac{1}{N\wc-(\Omega_\bq-q_\parallel v_\parallel)}
	\left[
		\left(
			J_{N+1}(z_\bq)\frac{\ee^{\ii\phi_\bq}}{\sqrt2}E_{\bq,1}
			+
			J_{N-1}(z_\bq)\frac{\ee^{-\ii\phi_\bq}}{\sqrt2}E_{\bq,-1}
		\right)
		\mathcal G_\bq
		+
		J_N(z_\bq)E_{\bq,0}
		\left(
			\partial_{v_\parallel}
			+
			\frac{N\wc}{v_\perp\Omega_\bq}\mathcal G
		\right)
	\right]f_0.
\end{dmath*}

\section{The nonlinear characteristic formula}\label{sec:nonlinear-characteristic}

The first nonlinear perturbation satisfies the same characteristic equation as the linear perturbation. For a triad \(\bk=\bp+\bq\), the nonlinear source is the field of the mode \(\bp\) acting on the linear response of the mode \(\bq\). The characteristic solution obtained in the standard calculation is
\begin{dmath*}
	f_\bk^{(2)}(\bv)
	=
	-
	\frac{e}{m\wc}
	\ee^{\ii z_\bk\sin\psi_\bk}
	\frac{1}{V}
	\sum_\bp{}'
	\sum_\bq{}'
	\delta_{\bq,\,\bk-\bp}
	\frac{1}{\Omega_\bp}
	\int_0^\infty
	\ee^{-\ii z_\bk\sin\psi_\bk(\lambda)}
	\ee^{-\ii\mu_\bk\lambda}
	\left.
		\left\{
			\ee^{\ii\psi_\bk(\lambda)}\mathcal P_{\bp,1}
			+
			\ee^{-\ii\psi_\bk(\lambda)}\mathcal P_{\bp,-1}
			+
			\mathcal P_{\bp,0}
		\right\}
		f_\bq^{(1)}
	\right|_{\bv(\lambda)}
	\dd\lambda
\end{dmath*}.
The differential blocks are independent of \(\lambda\). Their explicit form is:
\begin{dmath*}
	\mathcal P_{\bp,1}
	\defeq
	\ee^{\ii(\phi_\bk-\phi_\bp)}
	\left[
		\frac{\ee^{\ii\phi_\bp}}{\sqrt2}E_{\bp,1}
		\left(
			\Omega_\bp\mathcal G_\bp
			-
			\ii\frac{\mu_\bp\wc}{v_\perp}\partial_\phi
		\right)
		+
		E_{\bp,0}\frac{p_\perp}{2}
		\left(
			\mathcal G
			+
			\ii\frac{v_\parallel}{v_\perp}\partial_\phi
		\right)
	\right]
\end{dmath*};
\begin{dmath*}
	\mathcal P_{\bp,-1}
	\defeq
	\ee^{-\ii(\phi_\bk-\phi_\bp)}
	\left[
		\frac{\ee^{-\ii\phi_\bp}}{\sqrt2}E_{\bp,-1}
		\left(
			\Omega_\bp\mathcal G_\bp
			+
			\ii\frac{\mu_\bp\wc}{v_\perp}\partial_\phi
		\right)
		+
		E_{\bp,0}\frac{p_\perp}{2}
		\left(
			\mathcal G
			-
			\ii\frac{v_\parallel}{v_\perp}\partial_\phi
		\right)
	\right]
\end{dmath*};
\begin{dmath*}
	\mathcal P_{\bp,0}
	\defeq
	\ii p_\perp
	\frac{\ee^{-\ii\phi_\bp}E_{\bp,-1}-\ee^{\ii\phi_\bp}E_{\bp,1}}{\sqrt2}
	\partial_\phi
	+
	E_{\bp,0}\Omega_\bp\partial_{v_\parallel}
\end{dmath*}.
These differential blocks act on the full linear response $f_\bq^{(1)}$ and the resulting expression is evaluated at $\bv(\lambda)$. 
In particular, \(v_\perp\), \(v_\parallel\), \(z_\bq\), and \(\mu_\bq\) do not change; only the inner gyrophase acquires a dependence on $\lambda$ through
\begin{dmath*}
	\psi_\bq \mapsto {\psi_\bq(\lambda) = \psi_\bq + \lambda}
\end{dmath*}.

\section{Bivariate incomplete-Bessel functions}\label{sec:bivariate}

For later convenience, we anticipate the definition of the bivariate incomplete-Bessel function. A more detailed analysis of this function and the investigation of its main properties is deferred to \cref{sec:properties}.
\begin{definition}[Bivariate incomplete-Bessel function]
	For \(r\in\mathbb Z\), define
	\begin{dmath}\label{eq:defining_integral}
		G_{\mu,\nu}^{(r)}
		\left(
			z,\psi;
			w,\chi
		\right)
		\defeq
		\ii
		\int_0^\infty
		\ee^{-\ii z\sin(\psi+\lambda)}
		\ee^{\ii w\sin(\chi+\lambda)}
		\ee^{-\ii\mu\lambda}
		\ee^{\ii r\lambda}
		G_\nu
		\left(
			w,
			\chi+\lambda
		\right)
		\dd\lambda
	\end{dmath}.
\end{definition}

This function emerges as the natural orbit-resolvent kernel associated with the first nonlinear cyclotron response. The first triplet of variables \((z,\psi,\mu)\) belongs to the outer output mode. The second triplet \((w,\chi,\nu)\) belongs to the inner linear response. The exponential factor with positive sign is the inner Larmor factor already present in the full linear distribution. Therefore \(G_{\mu,\nu}^{(r)}\) is the exact orbit integral produced when the outer characteristic acts on a full inner incomplete-Bessel response.

The elementary replacement repeated used in the following is
\begin{dmath*}
	\int_0^\infty
	\ee^{-\ii z_\bk\sin(\psi_\bk+\lambda)}
	\ee^{\ii z_\bq\sin(\psi_\bq+\lambda)}
	\ee^{-\ii\mu_\bk\lambda}
	\ee^{\ii r\lambda}
	G_{\nu_\bq}
	\left(
		z_\bq,
		\psi_\bq+\lambda
	\right)
	\dd\lambda
	=
	-
	\ii
	G_{\mu_\bk,\nu_\bq}^{(r)}
	\left(
		z_\bk,\psi_\bk;
		z_\bq,\psi_\bq
	\right)
\end{dmath*}.
This identity is the nonlinear analogue of the linear characteristic evaluation. It contains the full inner Larmor factor and therefore also accounts for the phase-gradient terms produced when the nonlinear velocity-space derivatives act on the inner linear response.


\section{Nonlinear distribution in bivariate form}\label{sec:bivariate-distribution}


We write the full inner linear response as
\begin{dmath*}
	f_\bq^{(1)}
	=
	\frac{\ii e}{m\wc}
	\sum_{a\in\mathcal I}
	K_{\bq,a}	X_{\bq,a}
\end{dmath*},
where
\begin{dmath*}
	K_{\bq,a}\left(
		z_\bq,
		\psi_\bq
	\right)
	\defeq
	\ee^{\ii z_\bq\sin\psi_\bq+\ii a\psi_\bq}
	G_{\mu_\bq-a}
	\left(
		z_\bq,
		\psi_\bq
	\right)
\end{dmath*}.
Then the characteristic formula for the nonlinear response can be written in the compact form
\begin{dmath*}
	f_\bk^{(2)}(\bv)
	=
	-
	\frac{e^2}{m^2\wc^2}
	\ee^{\ii z_\bk\sin\psi_\bk}
	\frac{1}{V}
	\sum_\bp{}'
	\sum_\bq{}'
	\delta_{\bq,\,\bk-\bp}
	\frac{1}{\Omega_\bp} \cdot \\
	\ii
	\int_0^\infty
	\ee^{-\ii z_\bk\sin(\psi_\bk+\lambda)}
	\ee^{-\ii\mu_\bk\lambda}
	\sum_{a\in\mathcal I}
	\sum_{b\in\mathcal I}
	\left.
		\ee^{\ii b\psi_\bk}
		\ee^{\ii b\lambda}
		\mathcal P_{\bp,b}
		\left(
			K_{\bq,a}
			X_{\bq,a}
		\right)
	\right|_{\bv(\lambda)}
	\dd\lambda
\end{dmath*},
where 
\begin{dmath*}
	\mathcal I \defeq \{-1,0,1\}	
\end{dmath*}.
The index \(a\in\mathcal I\) labels the inner linear block, while the index \(b\in\mathcal I\) labels the outer gyrophase block generated by the nonlinear force.

\noindent The operator acting on the inner block is
\begin{dmath*}
	\mathcal P_{\bp,b}
	\defeq
	\ee^{\ii b(\phi_\bk-\phi_\bp)}
	\Bigg\{
		E_{\bp,b}
		\left[
			b^2
			\frac{\ee^{\ii b\phi_\bp}}{\sqrt2}
			\left(
				\Omega_\bp\mathcal G_\bp
				-
				\ii b
				\frac{\mu_\bp\wc}{v_\perp}
				\partial_\phi
			\right)
			+
			(1-b^2)
			\Omega_\bp
			\partial_{v_\parallel}
		\right]
		+
		\sum_{c=\pm1}
		E_{\bp,b+c}
		\frac{1-bc}{2}
		\left[
			b^2
			\frac{p_\perp}{2}
			\left(
				\mathcal G
				-
				\ii c
				\frac{v_\parallel}{v_\perp}
				\partial_\phi
			\right)
			-
			(1-b^2)
			\ii p_\perp c\sqrt2
			\ee^{\ii c\phi_\bp}
			\partial_\phi
		\right]
	\Bigg\}
\end{dmath*}.

Equivalently, by regrouping according to velocity derivatives,
\begin{dmath}\label{eq:P_operator}
	\mathcal P_{\bp,b}
	=
	\ee^{\ii b(\phi_\bk-\phi_\bp)}
	\left[
		C_{\bp,b}^{(\phi)}
		\partial_\phi
		+
		C_{\bp,b}^{(\perp)}
		\partial_{v_\perp}
		+
		C_{\bp,b}^{(\parallel)}
		\partial_{v_\parallel}
	\right]
\end{dmath},
where
\begin{dmath*}
	C_{\bp,b}^{(\phi)}
	\defeq
	-
	\ii b^3
	E_{\bp,b}
	\frac{\ee^{\ii b\phi_\bp}}{\sqrt2}
	\frac{\mu_\bp\wc}{v_\perp}
	+
	\sum_{c=\pm1}
	E_{\bp,b+c}
	\frac{1-bc}{2}
	\left[
		-
		\ii c b^2
		\frac{p_\perp v_\parallel}{2v_\perp}
		-
		(1-b^2)
		\ii p_\perp c\sqrt2
		\ee^{\ii c\phi_\bp}
	\right]
\end{dmath*},
\begin{dmath*}
	C_{\bp,b}^{(\perp)}
	\defeq
	b^2
	E_{\bp,b}
	\frac{\ee^{\ii b\phi_\bp}}{\sqrt2}
	\left(
		\Omega_\bp
		-
		p_\parallel v_\parallel
	\right)
	+
	\sum_{c=\pm1}
	E_{\bp,b+c}
	\frac{1-bc}{2}
	b^2
	\frac{p_\perp v_\parallel}{2}
\end{dmath*},
and
\begin{dmath*}
	C_{\bp,b}^{(\parallel)}
	\defeq
	E_{\bp,b}
	\left[
		b^2
		\frac{\ee^{\ii b\phi_\bp}}{\sqrt2}
		p_\parallel v_\perp
		+
		(1-b^2)\Omega_\bp
	\right]
	-
	\sum_{c=\pm1}
	E_{\bp,b+c}
	\frac{1-bc}{2}
	b^2
	\frac{p_\perp v_\perp}{2}
\end{dmath*}.
Values of \(E_{\bp,j}\) with \(j\notin\mathcal I\) are understood to vanish.  

The product rule gives
\begin{dmath}\label{eq:differential_block}
	\mathcal P_{\bp,b}
	\left(
		K_{\bq,a}
		X_{\bq,a}
	\right)
	=
	\left(
		\mathcal P_{\bp,b}
		K_{\bq,a}
	\right)
	X_{\bq,a}
	+
	K_{\bq,a}
	\left(
		\mathcal P_{\bp,b}
		X_{\bq,a}
	\right)
\end{dmath}.
Since
\begin{dmath*}
	{
		z_\bq
		=
		\frac{q_\perp v_\perp}{\wc},
		\qquad
		\psi_\bq
		=
		\phi-\phi_\bq,
		\qquad
		\mu_\bq
		=
		\frac{q_\parallel v_\parallel-\Omega_\bq}{\wc}
	}
\end{dmath*},
one has
\begin{dmath*}
	\partial_\phi K_{\bq,a}
	=
	\partial_{\psi_\bq}K_{\bq,a}
\end{dmath*},
\begin{dmath*}
	\partial_{v_\perp}K_{\bq,a}
	=
	\frac{q_\perp}{\wc}
	\partial_{z_\bq}K_{\bq,a}
\end{dmath*},
and
\begin{dmath*}
	\partial_{v_\parallel}K_{\bq,a}
	=
	\frac{q_\parallel}{\wc}
	\ee^{\ii z_\bq\sin\psi_\bq+\ii a\psi_\bq}
	\partial_{\mu_\bq}
	G_{\mu_\bq-a}
	\left(
		z_\bq,
		\psi_\bq
	\right)
\end{dmath*}.
Along the orbit
\begin{dmath*}
	\psi_\bq
	\mapsto
	\psi_\bq+\lambda
\end{dmath*},
whereas \(z_\bq\), \(\mu_\bq\), and \(X_{\bq,a}\) are constant.  Hence, grouping together the $\lambda$-dependent terms and using the definition of \(G_{\mu,\nu}^{(r)}\) in \cref{eq:defining_integral}, we obtain the replacement rules
\begin{dmath*}
	K_{\bq,a}
	\mapsto
	\ee^{\ii a\psi_\bq}
	G_{\mu_\bk,\mu_\bq-a}^{(a+b)}
	\left(
		z_\bk,\psi_\bk;
		z_\bq,\psi_\bq
	\right)
\end{dmath*},
\begin{dmath*}
	\partial_{\psi_\bq}K_{\bq,a}
	\mapsto
	\ee^{\ii a\psi_\bq}
	\left(
		\partial_{\psi_\bq}
		+
		\ii a
	\right)
	G_{\mu_\bk,\mu_\bq-a}^{(a+b)}
	\left(
		z_\bk,\psi_\bk;
		z_\bq,\psi_\bq
	\right)
\end{dmath*},
\begin{dmath*}
	\partial_{z_\bq}K_{\bq,a}
	\mapsto
	\ee^{\ii a\psi_\bq}
	\partial_{z_\bq}
	G_{\mu_\bk,\mu_\bq-a}^{(a+b)}
	\left(
		z_\bk,\psi_\bk;
		z_\bq,\psi_\bq
	\right)
\end{dmath*},
and
\begin{dmath*}
	\ee^{\ii z_\bq\sin\psi_\bq+\ii a\psi_\bq}
	\partial_{\mu_\bq}
	G_{\mu_\bq-a}
	\left(
		z_\bq,
		\psi_\bq
	\right)
	\mapsto
	\ee^{\ii a\psi_\bq}
	\left.
		\partial_\nu
		G_{\mu_\bk,\nu}^{(a+b)}
		\left(
			z_\bk,\psi_\bk;
			z_\bq,\psi_\bq
		\right)
	\right|_{\nu=\mu_\bq-a}
\end{dmath*}.
Therefore, by introducing the shortcut notations
\begin{dmath*}
	\mathbb G_{\bk,\bq,a,b}
	\defeq
	G_{\mu_\bk,\mu_\bq-a}^{(a+b)}
	\left(
		z_\bk,\psi_\bk;
		z_\bq,\psi_\bq
	\right)
\end{dmath*},
\begin{dmath*}
	\mathbb G_{\bk,\bq,a,b}^{(\phi)}
	\defeq
	\left(
		\partial_{\psi_\bq}
		+
		\ii a
	\right)
	G_{\mu_\bk,\mu_\bq-a}^{(a+b)}
	\left(
		z_\bk,\psi_\bk;
		z_\bq,\psi_\bq
	\right)
\end{dmath*},
\begin{dmath*}
	\mathbb G_{\bk,\bq,a,b}^{(\perp)}
	\defeq
	\partial_{z_\bq}
	G_{\mu_\bk,\mu_\bq-a}^{(a+b)}
	\left(
		z_\bk,\psi_\bk;
		z_\bq,\psi_\bq
	\right)
\end{dmath*},
and
\begin{dmath*}
	\mathbb G_{\bk,\bq,a,b}^{(\parallel)}
	\defeq
	\left.
		\partial_\nu
		G_{\mu_\bk,\nu}^{(a+b)}
		\left(
			z_\bk,\psi_\bk;
			z_\bq,\psi_\bq
		\right)
	\right|_{\nu=\mu_\bq-a}
\end{dmath*}.
the reduced nonlinear distribution can be eventually written as
\begin{dmath}\label{eq:clean-bivariate-response}
	\tilde f_\bk^{(2)}
	=
	-
	\frac{e^2}{m^2\wc^2}
	\frac{1}{V}
	\sum_\bp{}'
	\sum_\bq{}'
	\delta_{\bq,\,\bk-\bp}
	\frac{1}{\Omega_\bp}
	\sum_{a\in\mathcal I}
	\sum_{b\in\mathcal I}
	\ee^{\ii b\psi_\bk}
	\ee^{\ii b(\phi_\bk-\phi_\bp)}
	\ee^{\ii a\psi_\bq} \cdot \\
	\Bigg[
		\left(
			C_{\bp,b}^{(\phi)}
			\mathbb G_{\bk,\bq,a,b}^{(\phi)}
			+
			C_{\bp,b}^{(\perp)}
			\frac{q_\perp}{\wc}
			\mathbb G_{\bk,\bq,a,b}^{(\perp)}
			+
			C_{\bp,b}^{(\parallel)}
			\frac{q_\parallel}{\wc}
			\mathbb G_{\bk,\bq,a,b}^{(\parallel)}
		\right)
		X_{\bq,a}
		+
		\mathbb G_{\bk,\bq,a,b}
		\left[
			C_{\bp,b}^{(\perp)}
			\partial_{v_\perp}
			+
			C_{\bp,b}^{(\parallel)}
			\partial_{v_\parallel}
		\right]
		X_{\bq,a}
	\Bigg]
\end{dmath}.
Only the bivariate orders \(a+b\) occur.  Since \(a,b\in\{-1,0,1\}\), these orders range from \(-2\) to \(2\).  
\section{Properties of the bivariate functions}\label{sec:properties}

We now collect the identities needed to manipulate the nonlinear response and to compare the bivariate characteristic representation with the calssical double-harmonic formula.  The bivariate incomplete-Bessel function is defined by
\begin{dmath*}
	G_{\mu,\nu}^{(r)}
	\left(
		z,\psi;
		w,\chi
	\right)
	\defeq
	\ii
	\int_0^\infty
	\ee^{-\ii z\sin(\psi+\lambda)}
	\ee^{\ii w\sin(\chi+\lambda)}
	\ee^{-\ii\mu\lambda}
	\ee^{\ii r\lambda}
	G_\nu
	\left(
		w,
		\chi+\lambda
	\right)
	\dd\lambda
\end{dmath*}.
The same causal prescription used for the one-variable incomplete-Bessel function is understood.

It is useful to introduce the dressed inner factor
\begin{dmath}\label{eq:dressed_inner_factor}
	B_\nu(w,\chi)
	\defeq
	\ee^{\ii w\sin\chi}
	G_\nu(w,\chi)
\end{dmath}.
Its Fourier expansion is
\begin{dmath}\label{eq:FT_dressed_inner_factor}
	B_\nu(w,\chi)
	=
	\sum_{\ell=-\infty}^{\infty}
	T_\ell(w,\nu)
	\ee^{-\ii\ell\chi}
\end{dmath},
where
\begin{dmath*}
	T_\ell(w,\nu)
	\defeq
	\sum_{n=-\infty}^{\infty}
	\frac{
		J_n(w)J_{n-\ell}(w)
	}{
		n+\nu
	}
\end{dmath*}
is called Turkin's function in \citep{Newberger1982}.

\begin{proposition}[Superposition of one-variable incomplete-Bessel functions]
	The bivariate function admits the representation
	\begin{dmath}\label{eq:superposition_property}
		G_{\mu,\nu}^{(r)}
		\left(
			z,\psi;
			w,\chi
		\right)
		=
		\sum_{\ell=-\infty}^{\infty}
		T_\ell(w,\nu)
		\ee^{-\ii\ell\chi}
		G_{\mu+\ell-r}
		\left(
			z,
			\psi
		\right)
	\end{dmath}.
\end{proposition}

\begin{proof}
	Use the Fourier expansion of \(B_\nu\) at \(\chi+\lambda\) in the defining integral.  The remaining \(\lambda\)-integral is exactly \(G_{\mu+\ell-r}(z,\psi)\).
\end{proof}

\begin{proposition}[Harmonic expansion]
	The bivariate function has the double expansion
	\begin{dmath}\label{eq:bivariate_double_expansion}
		G_{\mu,\nu}^{(r)}
		\left(
			z,\psi;
			w,\chi
		\right)
		=
		\sum_{N=-\infty}^{\infty}
		\sum_{\ell=-\infty}^{\infty}
		\frac{
			J_N(z)
			T_\ell(w,\nu)
			\ee^{-\ii N\psi}
			\ee^{-\ii\ell\chi}
		}{
			N+\ell+\mu-r
		}
	\end{dmath}.
\end{proposition}

\begin{proof}
	This follows from the superposition formula and from
	\begin{dmath*}
		G_{\eta}(z,\psi)
		=
		\sum_{N=-\infty}^{\infty}
		\frac{
			J_N(z)
			\ee^{-\ii N\psi}
		}{
			N+\eta
		}
	\end{dmath*}.
\end{proof}

\begin{proposition}[Outer resolvent identity]
	The bivariate function satisfies
	\begin{dmath*}
		\left(
			\ii\partial_\psi
			+
			\ii\partial_\chi
			+
			\mu
			-
			r
		\right)
		G_{\mu,\nu}^{(r)}
		\left(
			z,\psi;
			w,\chi
		\right)
		=
		\ee^{-\ii z\sin\psi}
		B_\nu(w,\chi)
	\end{dmath*}.
\end{proposition}

\begin{proof}
	In the harmonic expansion the operator multiplies the term \((N,\ell)\) by \(N+\ell+\mu-r\), cancelling the denominator.
\end{proof}

\begin{proposition}[Outer Bessel equation]
	The bivariate function satisfies the homogeneous Bessel equation in the outer variables:
	\begin{dmath*}
		\left[
			z^2\partial_z^2
			+
			z\partial_z
			+
			z^2
			+
			\partial_\psi^2
		\right]
		G_{\mu,\nu}^{(r)}
		\left(
			z,\psi;
			w,\chi
		\right)
		=
		0
	\end{dmath*}.
\end{proposition}

\begin{proof}
	Each outer harmonic is proportional to \(J_N(z)\ee^{-\ii N\psi}\), which satisfies the corresponding Bessel equation.
\end{proof}

\begin{proposition}[Outer differential recurrence]
	The derivative with respect to the outer Larmor radius satisfies
	\begin{dmath*}
		\partial_z
		G_{\mu,\nu}^{(r)}
		\left(
			z,\psi;
			w,\chi
		\right)
		=
		\frac{1}{2}
		\left[
			\ee^{-\ii\psi}
			G_{\mu+1,\nu}^{(r)}
			\left(
				z,\psi;
				w,\chi
			\right)
			-
			\ee^{\ii\psi}
			G_{\mu-1,\nu}^{(r)}
			\left(
				z,\psi;
				w,\chi
			\right)
		\right]
	\end{dmath*}.
\end{proposition}

\begin{proof}
	Use the superposition formula and the corresponding one-variable recurrence.
\end{proof}

\begin{proposition}[Inner Turkin recurrence]
	The contracted inner coefficients satisfy
	\begin{dmath*}
		(\nu+\ell)T_\ell(w,\nu)
		+
		\frac{w}{2}
		\left[
			T_{\ell-1}(w,\nu)
			+
			T_{\ell+1}(w,\nu)
		\right]
		=
		\delta_{\ell,0}
	\end{dmath*}.
	Consequently,
	\begin{dmath*}
		\left(
			\ii\partial_\chi
			+
			\nu
		\right)
		G_{\mu,\nu}^{(r)}
		\left(
			z,\psi;
			w,\chi
		\right)
		+
		\frac{w}{2}
		\left[
			\ee^{\ii\chi}
			G_{\mu,\nu}^{(r+1)}
			\left(
				z,\psi;
				w,\chi
			\right)
			+
			\ee^{-\ii\chi}
			G_{\mu,\nu}^{(r-1)}
			\left(
				z,\psi;
				w,\chi
			\right)
		\right]
		=
		G_{\mu-r}(z,\psi)
	\end{dmath*}.
\end{proposition}

\begin{proof}
	The dressed factor satisfies
	\begin{dmath*}
		\left(
			\ii\partial_\chi
			+
			\nu
			+
			w\cos\chi
		\right)
		B_\nu(w,\chi)
		=
		1
	\end{dmath*}.
	Comparing Fourier coefficients gives the recurrence for \(T_\ell\).  The bivariate identity follows by multiplying by \(\ee^{-\ii\ell\chi}G_{\mu+\ell-r}(z,\psi)\) and summing over \(\ell\).
\end{proof}

\begin{proposition}[Phase insertions]
	For every \(s\in\mathbb Z\),
	\begin{dmath*}
		\ii
		\int_0^\infty
		\ee^{-\ii z\sin(\psi+\lambda)}
		\ee^{\ii w\sin(\chi+\lambda)}
		\ee^{-\ii\mu\lambda}
		\ee^{\ii r\lambda}
		\ee^{\ii s(\chi+\lambda)}
		G_\nu
		\left(
			w,
			\chi+\lambda
		\right)
		\dd\lambda
		=
		\ee^{\ii s\chi}
		G_{\mu,\nu}^{(r+s)}
		\left(
			z,\psi;
			w,\chi
		\right)
	\end{dmath*}.
\end{proposition}

\begin{proof}
	The factor \(\ee^{\ii s(\chi+\lambda)}\) gives \(\ee^{\ii s\chi}\) and shifts \(r\) to \(r+s\).
\end{proof}

\begin{proposition}[Inner derivative rules]
	The derivatives with respect to the inner variables act on the full dressed inner orbit factor:
	\begin{dmath*}
		\partial_w
		G_{\mu,\nu}^{(r)}
		=
		\ii
		\int_0^\infty
		\ee^{-\ii z\sin(\psi+\lambda)}
		\ee^{-\ii\mu\lambda}
		\ee^{\ii r\lambda}
		\partial_w
		\left[
			\ee^{\ii w\sin(\chi+\lambda)}
			G_\nu
			\left(
				w,
				\chi+\lambda
			\right)
		\right]
		\dd\lambda
	\end{dmath*},
	\begin{dmath*}
		\partial_\chi
		G_{\mu,\nu}^{(r)}
		=
		\ii
		\int_0^\infty
		\ee^{-\ii z\sin(\psi+\lambda)}
		\ee^{-\ii\mu\lambda}
		\ee^{\ii r\lambda}
		\partial_\chi
		\left[
			\ee^{\ii w\sin(\chi+\lambda)}
			G_\nu
			\left(
				w,
				\chi+\lambda
			\right)
		\right]
		\dd\lambda
	\end{dmath*},
	and
	\begin{dmath*}
		\partial_\nu
		G_{\mu,\nu}^{(r)}
		=
		\ii
		\int_0^\infty
		\ee^{-\ii z\sin(\psi+\lambda)}
		\ee^{\ii w\sin(\chi+\lambda)}
		\ee^{-\ii\mu\lambda}
		\ee^{\ii r\lambda}
		\partial_\nu
		G_\nu
		\left(
			w,
			\chi+\lambda
		\right)
		\dd\lambda
	\end{dmath*}.
\end{proposition}

\begin{proposition}[Double-characteristic representation]
	The bivariate function can also be written as
	\begin{dmath*}
		G_{\mu,\nu}^{(r)}
		\left(
			z,\psi;
			w,\chi
		\right)
		=
		-
		\int_0^\infty
		\int_0^\infty
		\ee^{-\ii z\sin(\psi+\lambda)}
		\ee^{\ii w\sin(\chi+\lambda)}
		\ee^{-\ii w\sin(\chi+\lambda+\sigma)}
		\ee^{-\ii(\mu-r)\lambda}
		\ee^{-\ii\nu\sigma}
		\dd\sigma
		\dd\lambda
	\end{dmath*}.
\end{proposition}

\begin{proof}
	Insert the one-variable characteristic representation of \(G_\nu\) into the definition of \(G_{\mu,\nu}^{(r)}\).
\end{proof}

\begin{proposition}[Degenerate one-orbit limit]
	At zero inner Larmor radius one has
	\begin{dmath*}
		G_{\mu,\nu}^{(r)}
		\left(
			z,\psi;
			0,\chi
		\right)
		=
		\frac{1}{\nu}
		G_{\mu-r}(z,\psi)
	\end{dmath*}.
\end{proposition}

\begin{proof}
	For \(w=0\), the inner dressed factor reduces to \(1/\nu\), and the defining integral becomes the one-variable incomplete-Bessel function with detuning \(\mu-r\).
\end{proof}

\section{Derivation of the double-harmonic formula}\label{sec:derivation}

We provide in the following a classical derivation of the double-harmonic formula for the nonlinear response, obtained by Fourier expanding the cyclic gyrophase-depending components before performing the orbit integral.  The result is a generalisation of the classical formula of Liu--Tripathi \citep{LiuTripathi1986} and will be used in the following for a consistency check of the correctness of the bivariate representation.

We start from the harmonic form of the inner linear response \citep{Ricci2026Bessel},
\begin{dmath*}
	f_\bq^{(1)}(\bv)
	=
	\ee^{\ii z_\bq\sin\psi_\bq}
	\sum_{M=-\infty}^{\infty}
	\ee^{-\ii M\psi_\bq}
	\tilde f_{\bq,M}^{(1)}
\end{dmath*}.
Equivalently, for each elementary block \(K_{\bq,a}X_{\bq,a}\), one may
write
\begin{dmath*}
	K_{\bq,a}
	X_{\bq,a}
	=
	\sum_{M=-\infty}^{\infty}
	\ee^{-\ii M\psi_\bq}
	K_{\bq,a,M}
	X_{\bq,a}
\end{dmath*}.
Since \(\partial_\phi=\partial_{\psi_\bq}\) on the inner response, the
action of the operator $\mathcal P_{\bp,b}$ in the form of \cref{eq:P_operator} on the \(M\)-th inner harmonic is
\begin{dmath*}
	\mathcal P_{\bp,b}
	\left(
		\ee^{-\ii M\psi_\bq}
		K_{\bq,a,M}
		X_{\bq,a}
	\right)
	=
	\ee^{\ii b(\phi_\bk-\phi_\bp)}
	\ee^{-\ii M\psi_\bq}
	\left[
		-\ii M C_{\bp,c}^{(\phi)}
		+
		C_{\bp,c}^{(\perp)}
		\partial_{v_\perp}
		+
		C_{\bp,c}^{(\parallel)}
		\partial_{v_\parallel}
	\right]
	\left(
		K_{\bq,a,M}
		X_{\bq,a}
	\right)
\end{dmath*}.
This is the essential reduction.  The gyrophase derivative in the bivariate
formula becomes multiplication by \(-\ii M\) in the harmonic formula.

Along the outer orbit,
\begin{dmath*}
	\psi_\bq
	\mapsto
	\psi_\bq+\lambda
\end{dmath*},
and the outer block labelled by \(b\) contributes the phase
\begin{dmath*}
	\ee^{\ii b\psi_\bk}
	\ee^{\ii b\lambda}
	\ee^{-\ii M(\psi_\bq+\lambda)}
	=
	\ee^{\ii b\psi_\bk}
	\ee^{-\ii M\psi_\bq}
	\ee^{-\ii(M-b)\lambda}
\end{dmath*}.
The remaining outer characteristic integral therefore selects the harmonic
denominator with index \(N=M-b\), or equivalently the Bessel factor
\(J_{N-M+b}(z_\bp)\) in the standard notation.  Thus the three values of
\(b\) give the three blocks:
\begin{align*}
		b = 1
		& \Longrightarrow 
		J_{N-M+1}(z_\bp) \\
		b = -1
		& \Longrightarrow 
		J_{N-M-1}(z_\bp) \\
		b = 0
		& \Longrightarrow 
		J_{N-M}(z_\bp)
\end{align*}

It remains only to evaluate the three coefficient combinations.  For \(b=1\)
one has
\begin{dmath*}
	-\ii M C_{\bp,1}^{(\phi)}
	+
	C_{\bp,1}^{(\perp)}
	\partial_{v_\perp}
	+
	C_{\bp,1}^{(\parallel)}
	\partial_{v_\parallel}
	=
	\frac{\ee^{\ii\phi_\bp}}{\sqrt2}E_{\bp,1}
	\left(
		\Omega_\bp\mathcal G_\bp
		-
		M\frac{\wc\mu_\bp}{v_\perp}
	\right)
	+
	\frac{p_\perp E_{\bp,0}}{2}
	\left(
		\mathcal G
		+
		M\frac{v_\parallel}{v_\perp}
	\right)
\end{dmath*}.
For \(b=-1\),
\begin{dmath*}
	-\ii M C_{\bp,-1}^{(\phi)}
	+
	C_{\bp,-1}^{(\perp)}
	\partial_{v_\perp}
	+
	C_{\bp,-1}^{(\parallel)}
	\partial_{v_\parallel}
	=
	\frac{\ee^{-\ii\phi_\bp}}{\sqrt2}E_{\bp,-1}
	\left(
		\Omega_\bp\mathcal G_\bp
		+
		M\frac{\wc\mu_\bp}{v_\perp}
	\right)
	+
	\frac{p_\perp E_{\bp,0}}{2}
	\left(
		\mathcal G
		-
		M\frac{v_\parallel}{v_\perp}
	\right)
\end{dmath*}.
Finally, for \(b=0\),
\begin{dmath*}
	-\ii M C_{\bp,0}^{(\phi)}
	+
	C_{\bp,0}^{(\perp)}
	\partial_{v_\perp}
	+
	C_{\bp,0}^{(\parallel)}
	\partial_{v_\parallel}
	=
	M
	\frac{p_\perp}{\sqrt2}
	\left(
		\ee^{-\ii\phi_\bp}E_{\bp,-1}
		-
		\ee^{\ii\phi_\bp}E_{\bp,1}
	\right)
	+
	\Omega_\bp E_{\bp,0}
	\partial_{v_\parallel}
\end{dmath*}.
Combining these three cases gives
\begin{dmath*}
	f_\bk^{(2)}(\bv)
	=
	\ee^{\ii z_\bk\sin\psi_\bk}
	\sum_{N=-\infty}^{\infty}
	\ee^{-\ii N\psi_\bk}
	\tilde f_{\bk,N}^{(2)}
\end{dmath*},
with
\begin{dmath*}
	\tilde f_{\bk,N}^{(2)}
	=
	\ii\frac{q}{m}
	\frac{1}{N\wc-(\Omega_\bk-k_\parallel v_\parallel)}
	\frac{1}{V}
	\sum_\bp{}'
	\sum_\bq{}'
	\delta_{\bq,\,\bk-\bp}
	\frac{1}{\Omega_\bp} \cdot \\
	\sum_{M=-\infty}^{\infty}
	\ee^{-\ii\left[N\phi_\bk-(N-M)\phi_\bp-M\phi_\bq\right]}
	\mathcal L_{NM}^{(\bp)}
	\tilde f_{\bq,M}^{(1)}
\end{dmath*},
where
\begin{dmath*}
	\mathcal L_{NM}^{(\bp)}
	=
	\mathcal L_{NM,1}^{(\bp)}
	+
	\mathcal L_{NM,-1}^{(\bp)}
	+
	\mathcal L_{NM,0}^{(\bp)}
\end{dmath*},
and
\begin{dmath*}
	\mathcal L_{NM,1}^{(\bp)}
	=
	J_{N-M+1}(z_\bp)
	\left[
		\frac{\ee^{\ii\phi_\bp}}{\sqrt2}E_{\bp,1}
		\left(
			\Omega_\bp\mathcal G_\bp
			-
			M\frac{\wc\mu_\bp}{v_\perp}
		\right)
		+
		\frac{p_\perp E_{\bp,0}}{2}
		\left(
			\mathcal G
			+
			M\frac{v_\parallel}{v_\perp}
		\right)
	\right]
\end{dmath*},
\begin{dmath*}
	\mathcal L_{NM,-1}^{(\bp)}
	=
	J_{N-M-1}(z_\bp)
	\left[
		\frac{\ee^{-\ii\phi_\bp}}{\sqrt2}E_{\bp,-1}
		\left(
			\Omega_\bp\mathcal G_\bp
			+
			M\frac{\wc\mu_\bp}{v_\perp}
		\right)
		+
		\frac{p_\perp E_{\bp,0}}{2}
		\left(
			\mathcal G
			-
			M\frac{v_\parallel}{v_\perp}
		\right)
	\right]
\end{dmath*},
\begin{dmath*}
	\mathcal L_{NM,0}^{(\bp)}
	=
	J_{N-M}(z_\bp)
	\left[
		\Omega_\bp E_{\bp,0}
		\partial_{v_\parallel}
		+
		M
		\frac{p_\perp}{\sqrt2}
		\left(
			\ee^{-\ii\phi_\bp}E_{\bp,-1}
			-
			\ee^{\ii\phi_\bp}E_{\bp,1}
		\right)
	\right]
\end{dmath*}.
This is the classical double-harmonic formula. In the next section we show that this formula can be directly obtained from \cref{eq:clean-bivariate-response} using the properties of the bivariate incomplete Bessel functions discussed in \cref{sec:properties}.

\section{Recovery of the traditional double-harmonic formula}\label{sec:recovery}

We now verify that the bivariate response reproduces the double-harmonic expression.  The point of the verification is to start from \cref{eq:clean-bivariate-response}, use only the identities of \cref{sec:properties}, and recover the formula independently obtained in \cref{sec:derivation}.

We first use the superposition property of the bivariate function, \cref{eq:superposition_property}.  With
\(\nu=\mu_\bq-a\), \(r=a+b\), and \(\ell=m+a\), it gives
\begin{dmath*}
	\ee^{\ii a\psi_\bq}
	G_{\mu_\bk,\mu_\bq-a}^{(a+b)}
	\left(
		z_\bk,\psi_\bk;
		z_\bq,\psi_\bq
	\right)
	=
	\sum_{m=-\infty}^{\infty}
	\ee^{-\ii m\psi_\bq}
	T_{m+a}
	\left(
		z_\bq,
		\mu_\bq-a
	\right)
	G_{\mu_\bk+m-b}
	\left(
		z_\bk,
		\psi_\bk
	\right)
\end{dmath*}.
The differentiated kernels are expanded in the same way:
\begin{dmath*}
	\ee^{\ii a\psi_\bq}
	\left(
		\partial_{\psi_\bq}
		+
		\ii a
	\right)
	G_{\mu_\bk,\mu_\bq-a}^{(a+b)}
	=
	\sum_{m=-\infty}^{\infty}
	\ee^{-\ii m\psi_\bq}
	\left(
		-\ii m
	\right)
	T_{m+a}
	\left(
		z_\bq,
		\mu_\bq-a
	\right)
	G_{\mu_\bk+m-b}
	\left(
		z_\bk,
		\psi_\bk
	\right)
\end{dmath*},
\begin{dmath*}
	\ee^{\ii a\psi_\bq}
	\partial_{z_\bq}
	G_{\mu_\bk,\mu_\bq-a}^{(a+b)}
	=
	\sum_{m=-\infty}^{\infty}
	\ee^{-\ii m\psi_\bq}
	\partial_{z_\bq}
	T_{m+a}
	\left(
		z_\bq,
		\mu_\bq-a
	\right)
	G_{\mu_\bk+m-b}
	\left(
		z_\bk,
		\psi_\bk
	\right)
\end{dmath*},
and
\begin{dmath*}
	\ee^{\ii a\psi_\bq}
	\left.
		\partial_\nu
		G_{\mu_\bk,\nu}^{(a+b)}
	\right|_{\nu=\mu_\bq-a}
	=
	\sum_{m=-\infty}^{\infty}
	\ee^{-\ii m\psi_\bq}
	\partial_{\mu_\bq}
	T_{m+a}
	\left(
		z_\bq,
		\mu_\bq-a
	\right)
	G_{\mu_\bk+m-b}
	\left(
		z_\bk,
		\psi_\bk
	\right)
\end{dmath*}.
Thus the sum over \(a\) in \cref{eq:clean-bivariate-response} reconstructs
the \(m\)-th harmonic of the full inner linear response.  More explicitly, if
\begin{dmath*}
	Y_{\bq,m}
	\defeq
	\sum_{a\in\mathcal I}
	T_{m+a}
	\left(
		z_\bq,
		\mu_\bq-a
	\right)
	X_{\bq,a}
\end{dmath*},
then
\begin{dmath*}
	f_{\bq,m}^{(1)}
	=
	\frac{\ii e}{m\wc}
	Y_{\bq,m}
\end{dmath*}.
The differentiated bivariate blocks reconstruct the corresponding derivatives
of \(Y_{\bq,m}\).  Consequently the bracket in
\cref{eq:clean-bivariate-response} becomes
\begin{dmath*}
	G_{\mu_\bk+m-b}
	\left(
		z_\bk,
		\psi_\bk
	\right)
	\mathcal D_{\bp,b}^{(m)}
	Y_{\bq,m}
\end{dmath*},
where
\begin{dmath*}
	\mathcal D_{\bp,b}^{(m)}
	\defeq
	-\ii m C_{\bp,b}^{(\phi)}
	+
	C_{\bp,b}^{(\perp)}
	\partial_{v_\perp}
	+
	C_{\bp,b}^{(\parallel)}
	\partial_{v_\parallel}
\end{dmath*}.
Using \(Y_{\bq,m}=m\wc f_{\bq,m}^{(1)}/(\ii e)\), the reduced bivariate
response becomes
\begin{dmath*}
	\tilde f_\bk^{(2)}
	=
	\frac{\ii e}{m\wc}
	\frac{1}{V}
	\sum_\bp{}'
	\sum_\bq{}'
	\delta_{\bq,\,\bk-\bp}
	\frac{1}{\Omega_\bp} \cdot \\
	\sum_{b\in\mathcal I}
	\sum_{m=-\infty}^{\infty}
	\ee^{\ii b\psi_\bk}
	\ee^{\ii b(\phi_\bk-\phi_\bp)}
	\ee^{-\ii m\psi_\bq}
	G_{\mu_\bk+m-b}
	\left(
		z_\bk,
		\psi_\bk
	\right)
	\mathcal D_{\bp,b}^{(m)}
	f_{\bq,m}^{(1)}
\end{dmath*}.
Since
\begin{dmath*}
	\psi_\bq
	=
	\psi_\bk+\phi_\bk-\phi_\bq
\end{dmath*},
the phase factor satisfies
\begin{dmath*}
	\ee^{\ii b\psi_\bk}
	\ee^{-\ii m\psi_\bq}
	=
	\ee^{-\ii(m-b)\psi_\bk}
	\ee^{-\ii m(\phi_\bk-\phi_\bq)}
\end{dmath*}.
We now expand the remaining one-variable incomplete-Bessel function:
\begin{dmath*}
	G_{\mu_\bk+m-b}
	\left(
		z_\bk,
		\psi_\bk
	\right)
	=
	\sum_{P=-\infty}^{\infty}
	\ee^{-\ii P\psi_\bk}
	\frac{
		J_P(z_\bk)
	}{
		P+\mu_\bk+m-b
	}
\end{dmath*}.
Setting
\begin{dmath*}
	N
	\defeq
	P+m-b
\end{dmath*},
or equivalently \(P=N-m+b\), gives
\begin{dmath*}
	P+\mu_\bk+m-b
	=
	N+\mu_\bk
\end{dmath*}.
Therefore
\begin{dmath*}
	\frac{1}{
		P+\mu_\bk+m-b
	}
	=
	\frac{\wc}{
		N\wc-\left(
			\Omega_\bk-k_\parallel v_\parallel
		\right)
	}
\end{dmath*}.
At this stage the bivariate formula gives the intermediate harmonic form
\begin{dmath*}
	\tilde f_{\bk,N}^{(2)}
	=
	\ii\frac{e}{m}
	\frac{1}{
		N\wc-\left(
			\Omega_\bk-k_\parallel v_\parallel
		\right)
	}
	\frac{1}{V}
	\sum_\bp{}'
	\sum_\bq{}'
	\delta_{\bq,\,\bk-\bp}
	\frac{1}{\Omega_\bp} \cdot \\
	\sum_{b\in\mathcal I}
	\sum_{m=-\infty}^{\infty}
	J_{N-m+b}(z_\bk)
	\ee^{-\ii m(\phi_\bk-\phi_\bq)}
	\ee^{\ii b(\phi_\bk-\phi_\bp)}
	\mathcal D_{\bp,b}^{(m)}
	f_{\bq,m}^{(1)}
\end{dmath*}.
This formula is already the harmonic expansion of the bivariate response.
It remains only to rewrite the full inner harmonic \(f_{\bq,m}^{(1)}\) in
terms of the reduced inner harmonic coefficients used in the traditional
formula.  By definition,
\begin{dmath*}
	f_\bq^{(1)}
	=
	\ee^{\ii z_\bq\sin\psi_\bq}
	\sum_{M=-\infty}^{\infty}
	\ee^{-\ii M\psi_\bq}
	\tilde f_{\bq,M}^{(1)}
\end{dmath*},
and therefore
\begin{dmath*}
	f_{\bq,m}^{(1)}
	=
	\sum_{M=-\infty}^{\infty}
	J_{M-m}(z_\bq)
	\tilde f_{\bq,M}^{(1)}
\end{dmath*}.
The differential operator \(\mathcal D_{\bp,b}^{(m)}\) acts on the product
\(J_{M-m}(z_\bq)\tilde f_{\bq,M}^{(1)}\).  The only identities needed to
move the derivatives through the Bessel factor are
\begin{dmath*}
	\Omega_\bp\mathcal G_\bp J_{M-m}(z_\bq)
	=
	-\mu_\bp q_\perp J'_{M-m}(z_\bq)
\end{dmath*},
\begin{dmath*}
	\mathcal G J_{M-m}(z_\bq)
	=
	\frac{v_\parallel q_\perp}{\wc}
	J'_{M-m}(z_\bq)
\end{dmath*},
and
\begin{dmath*}
	J'_s(z_\bq)
	=
	\frac{1}{2}
	\left[
		J_{s-1}(z_\bq)
		-
		J_{s+1}(z_\bq)
	\right]
\end{dmath*}.
After these elementary reductions, all terms contain sums of the form
\begin{dmath*}
	\sum_{m=-\infty}^{\infty}
	J_{N-m+b}(z_\bk)
	J_{M-m}(z_\bq)
	\ee^{-\ii m(\phi_\bk-\phi_\bq)}
	\ee^{\ii b(\phi_\bk-\phi_\bp)}
\end{dmath*}.
The remaining sum over the intermediate harmonic \(m\) is evaluated by
Graf's addition theorem.  With the phase convention used here and
\(\bk=\bp+\bq\), this theorem may be written as
\begin{dmath*}
	\ee^{\ii l(\phi_\bp-\phi_\bk)}
	J_l(z_\bp)
	=
	\sum_{h=-\infty}^{\infty}
	J_{l+h}(z_\bk)
	J_h(z_\bq)
	\ee^{\ii h(\phi_\bk-\phi_\bq)}
\end{dmath*},
see, for instance, \cite[Eq.~10.23.7]{NIST:DLMF}.  Taking
\begin{dmath*}
	l
	= {
	N-M+b,
	\qquad
	h
	=
	M-m
	}
\end{dmath*},
one obtains
\begin{dmath*}
	\sum_{m=-\infty}^{\infty}
	J_{N-m+b}(z_\bk)
	J_{M-m}(z_\bq)
	\ee^{-\ii m(\phi_\bk-\phi_\bq)}
	\ee^{\ii b(\phi_\bk-\phi_\bp)}
	=
	J_{N-M+b}(z_\bp)
	\ee^{-\ii\left[
		N\phi_\bk
		-
		(N-M)\phi_\bp
		-
		M\phi_\bq
	\right]}
\end{dmath*}.
This is the precise point at which the \(z_\bp\)-Bessel factor of the
classical formula appears.  It is produced by the Graf contraction over the
intermediate full harmonic \(m\).

Evaluating the three possible values of \(b\) now gives the three blocks of
the classical formula.  For \(b=1\),
\begin{dmath*}
	\mathcal L_{NM,1}^{(\bp)}
	=
	J_{N-M+1}(z_\bp)
	\left[
		\frac{\ee^{\ii\phi_\bp}}{\sqrt2}
		E_{\bp,1}
		\left(
			\Omega_\bp\mathcal G_\bp
			-
			M
			\frac{\wc\mu_\bp}{v_\perp}
		\right)
		+
		\frac{p_\perp E_{\bp,0}}{2}
		\left(
			\mathcal G
			+
			M
			\frac{v_\parallel}{v_\perp}
		\right)
	\right]
\end{dmath*}.
For \(b=-1\),
\begin{dmath*}
	\mathcal L_{NM,-1}^{(\bp)}
	=
	J_{N-M-1}(z_\bp)
	\left[
		\frac{\ee^{-\ii\phi_\bp}}{\sqrt2}
		E_{\bp,-1}
		\left(
			\Omega_\bp\mathcal G_\bp
			+
			M
			\frac{\wc\mu_\bp}{v_\perp}
		\right)
		+
		\frac{p_\perp E_{\bp,0}}{2}
		\left(
			\mathcal G
			-
			M
			\frac{v_\parallel}{v_\perp}
		\right)
	\right]
\end{dmath*}.
For \(b=0\),
\begin{dmath*}
	\mathcal L_{NM,0}^{(\bp)}
	=
	J_{N-M}(z_\bp)
	\left[
		\Omega_\bp E_{\bp,0}
		\partial_{v_\parallel}
		+
		M
		\frac{p_\perp}{\sqrt2}
		\left(
			\ee^{-\ii\phi_\bp}E_{\bp,-1}
			-
			\ee^{\ii\phi_\bp}E_{\bp,1}
		\right)
	\right]
\end{dmath*}.
Thus
\begin{dmath*}
	\mathcal L_{NM}^{(\bp)}
	=
	\mathcal L_{NM,1}^{(\bp)}
	+
	\mathcal L_{NM,-1}^{(\bp)}
	+
	\mathcal L_{NM,0}^{(\bp)}
\end{dmath*}.
Substitution gives
\begin{dmath*}
	f_\bk^{(2)}(\bv)
	=
	\ee^{\ii z_\bk\sin\psi_\bk}
	\sum_{N=-\infty}^{\infty}
	\ee^{-\ii N\psi_\bk}
	\tilde f_{\bk,N}^{(2)}
\end{dmath*},
with
\begin{dmath*}
	\tilde f_{\bk,N}^{(2)}
	=
	\ii\frac{e}{m}
	\frac{1}{
		N\wc-\left(
			\Omega_\bk-k_\parallel v_\parallel
		\right)
	}
	\frac{1}{V}
	\sum_\bp{}'
	\sum_\bq{}'
	\delta_{\bq,\,\bk-\bp}
	\frac{1}{\Omega_\bp} \cdot \\
	\sum_{M=-\infty}^{\infty}
	\ee^{-\ii\left[
		N\phi_\bk
		-
		(N-M)\phi_\bp
		-
		M\phi_\bq
	\right]}
	\mathcal L_{NM}^{(\bp)}
	\tilde f_{\bq,M}^{(1)}
\end{dmath*}.
This is exactly the traditional double-harmonic formula derived independently
in \cref{sec:derivation}.  The proof shows explicitly that the
bivariate formula is not a different approximation: it is the same harmonic
response before the final expansion of the inner full harmonic and the Graf
contraction over the intermediate index \(m\).

\section{Projected kernels and nonlinear susceptibility}\label{sec:susceptibility}

The bivariate representation obtained above gives a compact expression for the second-order distribution function.  The nonlinear susceptibility tensor is obtained by projecting this distribution on the velocity moments entering the current.  In this section we perform this projection without returning to the double-harmonic expansion.  The result is a susceptibility written directly in terms of projected bivariate kernels.  This form keeps the orbit structure of the nonlinear response explicit and isolates the angular contractions into a small number of reusable special-function blocks.

The second-order nonlinear plasma susceptibility is the third-rank tensor implicitly defined by the relation  \citep{Sitenko1982}:
\begin{dmath*}
	\frac{1}{V}\sum_\bp{}^'\sum_\bq{}^'\delta_{\bq, \bk-\bp}\,{\chi}^{(2)}_{ijk}(\bp, \Omega_\bp; \bq, \Omega_\bq)\,E_{\bp,j}\,E_{\bq,k}= {\frac{\ii}{\epsilon_0 \Omega_\bk}\,j^{(2)}_{\bk,i} }
\end{dmath*},
where Einstein's convention is adopted for repeated indices.
The determination of the nonlinear susceptibility requires the calculation of the second-order total nonlinear current density.

For a single species the nonlinear current is
\begin{dmath*}
	\bj_\bk^{(2)}
	=
	e
	\int
	\bv
	f_\bk^{(2)}(\bv)
	\dd^3v
\end{dmath*}.
Since
\begin{dmath}\label{eq:v_circular}
	\bv
	=
	\frac{v_\perp}{\sqrt2}
	\left(
		\ee^{-\ii\phi}\mathbf e_L
		+
		\ee^{\ii\phi}\mathbf e_R
	\right)
	+
	v_\parallel\mathbf b
\end{dmath},
the angular integral selects only fixed gyrophase Fourier components of the full distribution.
Thus the susceptibility calculation only requires angular projections of the bivariate kernels and of their derivatives.

Let
\begin{dmath*}
	\Delta
	\defeq
	\phi_\bk-\phi_\bq
\end{dmath*}.
Along the output angular integration one has
\begin{dmath*}
	\psi_\bq
	=
	\psi_\bk-\Delta
\end{dmath*}.
For the elementary block labelled by \(a\in\mathcal I\) and \(b\in\mathcal I\), we define the projected bivariate kernel
\begin{dmath*}
	\mathcal K_N^{(r,a)}
	\left(
		z_\bk,z_\bq,\Delta;
		\mu_\bk,\nu
	\right)
	\defeq
	\frac{1}{2\pi}
	\int_0^{2\pi}
	\ee^{\ii N\psi_\bk}
	\ee^{\ii z_\bk\sin\psi_\bk}
	\ee^{\ii a(\psi_\bk-\Delta)}
	G_{\mu_\bk,\nu}^{(r)}
	\left(
		z_\bk,\psi_\bk;
		z_\bq,\psi_\bk-\Delta
	\right)
	\dd\psi_\bk
\end{dmath*}.
This is the angular Fourier projection of exactly the factor that appears in the current moment.
In the calculation of the susceptibility we have specifically
\begin{dmath*}
	{
		r=a+b,
		\qquad
		\nu=\mu_\bq-a
	}
\end{dmath*}.
The integer \(N\) is the output harmonic selected by the current projection. With the convention used in the definition of \(\mathcal K_N^{(r,a)}\), the angular factors coming from the velocity projection are \(N_1=-1\), \(N_{-1}=1\), and \(N_0=0\). Equivalently, the left and right current components select the \(-1\) and \(+1\) harmonics of the full distribution, respectively.

Using the superposition formula
\begin{dmath*}
	G_{\mu,\nu}^{(r)}
	\left(
		z,\psi;
		w,\chi
	\right)
	=
	\sum_{\ell=-\infty}^{\infty}
	T_\ell(w,\nu)
	\ee^{-\ii\ell\chi}
	G_{\mu+\ell-r}
	\left(
		z,
		\psi
	\right)
\end{dmath*},
one obtains the coefficient representation
\begin{dmath*}
	\mathcal K_N^{(r,a)}
	\left(
		z_\bk,z_\bq,\Delta;
		\mu_\bk,\nu
	\right)
	=
	\sum_{\ell=-\infty}^{\infty}
	\ee^{\ii(\ell-a)\Delta}
	T_\ell(z_\bq,\nu)
	T_{N+a-\ell}
	\left(
		z_\bk,
		\mu_\bk+\ell-r
	\right)
\end{dmath*},
where the one-variable projected coefficient is
\begin{dmath*}
	T_M(z,\eta)
	\defeq
	\frac{1}{2\pi}
	\int_0^{2\pi}
	\ee^{\ii M\psi}
	\ee^{\ii z\sin\psi}
	G_{\eta}
	\left(
		z,
		\psi
	\right)
	\dd\psi
\end{dmath*}.
Equivalently, after setting \(n=\ell-a\),
\begin{dmath*}
	\mathcal K_N^{(r,a)}
	\left(
		z_\bk,z_\bq,\Delta;
		\mu_\bk,\nu
	\right)
	=
	\sum_{n=-\infty}^{\infty}
	\ee^{\ii n\Delta}
	T_{n+a}(z_\bq,\nu)
	T_{N-n}
	\left(
		z_\bk,
		\mu_\bk+n+a-r
	\right)
\end{dmath*}.
This is the projected analogue of the bivariate function. It is the object that replaces the full two-angle kernel in the current moment.

The derivative blocks appearing in the nonlinear distribution have equally simple projections. Since
\begin{dmath*}
	\psi_\bq
	=
	\psi_\bk-\Delta
\end{dmath*},
differentiation with respect to the inner gyrophase at fixed \(\psi_\bk\) is equivalent to
\begin{dmath*}
	\partial_{\psi_\bq}
	=
	-
	\partial_\Delta
\end{dmath*}.
Moreover,
\begin{dmath*}
	\ee^{\ii a\psi_\bq}
	\left(
		\partial_{\psi_\bq}
		+
		\ii a
	\right)
	G_{\mu_\bk,\nu}^{(r)}
	=
	\partial_{\psi_\bq}
	\left[
		\ee^{\ii a\psi_\bq}
		G_{\mu_\bk,\nu}^{(r)}
	\right]
\end{dmath*}.
Therefore its projected form is
\begin{dmath*}
	\frac{1}{2\pi}
	\int_0^{2\pi}
	\ee^{\ii N\psi_\bk}
	\ee^{\ii z_\bk\sin\psi_\bk}
	\ee^{\ii a(\psi_\bk-\Delta)}
	\left(
		\partial_{\psi_\bq}
		+
		\ii a
	\right)
	G_{\mu_\bk,\nu}^{(r)}
	\left(
		z_\bk,\psi_\bk;
		z_\bq,\psi_\bk-\Delta
	\right)
	\dd\psi_\bk
	=
	-
	\partial_\Delta
	\mathcal K_N^{(r,a)}
\end{dmath*}.
Similarly,
\begin{dmath*}
	\frac{1}{2\pi}
	\int_0^{2\pi}
	\ee^{\ii N\psi_\bk}
	\ee^{\ii z_\bk\sin\psi_\bk}
	\ee^{\ii a(\psi_\bk-\Delta)}
	\partial_{z_\bq}
	G_{\mu_\bk,\nu}^{(r)}
	\left(
		z_\bk,\psi_\bk;
		z_\bq,\psi_\bk-\Delta
	\right)
	\dd\psi_\bk
	=
	\partial_{z_\bq}
	\mathcal K_N^{(r,a)}
\end{dmath*},
and
\begin{dmath*}
	\frac{1}{2\pi}
	\int_0^{2\pi}
	\ee^{\ii N\psi_\bk}
	\ee^{\ii z_\bk\sin\psi_\bk}
	\ee^{\ii a(\psi_\bk-\Delta)}
	\partial_\nu
	G_{\mu_\bk,\nu}^{(r)}
	\left(
		z_\bk,\psi_\bk;
		z_\bq,\psi_\bk-\Delta
	\right)
	\dd\psi_\bk
	=
	\partial_\nu
	\mathcal K_N^{(r,a)}
\end{dmath*}.
Thus all angular contractions required by the nonlinear current are finite combinations of
\begin{dmath*}
	{
		\mathcal K_N^{(r,a)},
		\qquad
		\partial_\Delta\mathcal K_N^{(r,a)},
		\qquad
		\partial_{z_\bq}\mathcal K_N^{(r,a)},
		\qquad
		\partial_\nu\mathcal K_N^{(r,a)}
	}
\end{dmath*}.
We note that the derivative with respect to the relative angle may also be written in terms of the multiplier operator
\begin{dmath*}
	\mathfrak D_\Delta
	\defeq
	-
	\ii\partial_\Delta
\end{dmath*},
so that
\begin{dmath*}
	\mathfrak D_\Delta
	\mathcal K_N^{(r,a)}
	=
	\sum_{n=-\infty}^{\infty}
	n
	\ee^{\ii n\Delta}
	T_{n+a}(z_\bq,\nu)
	T_{N-n}
	\left(
		z_\bk,
		\mu_\bk+n+a-r
	\right)
\end{dmath*}.

The bivariate representation gives the nonlinear current by substituting the reduced distribution of \cref{sec:bivariate-distribution} into the current moment and applying the projected kernels above. More explicitly, each occurrence of
\begin{dmath*}
	\ee^{\ii a\psi_\bq}
	G_{\mu_\bk,\mu_\bq-a}^{(a+b)}
	\left(
		z_\bk,\psi_\bk;
		z_\bq,\psi_\bq
	\right)
\end{dmath*}
is replaced, in the \(N\)-th current projection, by
\begin{dmath*}
	\mathcal K_N^{(a+b,a)}
	\left(
		z_\bk,z_\bq,\Delta;
		\mu_\bk,\mu_\bq-a
	\right)
\end{dmath*}.
Likewise,
\begin{dmath*}
	\ee^{\ii a\psi_\bq}
	\left(
		\partial_{\psi_\bq}
		+
		\ii a
	\right)
	G_{\mu_\bk,\mu_\bq-a}^{(a+b)}
\end{dmath*}
is replaced by
\begin{dmath*}
	-
	\partial_\Delta
	\mathcal K_N^{(a+b,a)}
	\left(
		z_\bk,z_\bq,\Delta;
		\mu_\bk,\mu_\bq-a
	\right)
\end{dmath*},
while the \(z_\bq\)- and \(\nu\)-derivative terms are replaced by
\begin{dmath*}
	\partial_{z_\bq}
	\mathcal K_N^{(a+b,a)}
	\left(
		z_\bk,z_\bq,\Delta;
		\mu_\bk,\mu_\bq-a
	\right)
\end{dmath*}
and
\begin{dmath*}
	\left.
	\partial_\nu
	\mathcal K_N^{(a+b,a)}
	\left(
		z_\bk,z_\bq,\Delta;
		\mu_\bk,\nu
	\right)
	\right|_{\nu=\mu_\bq-a}
\end{dmath*}.
Consequently the nonlinear current is built from the projected family
\begin{dmath*}
{
	\mathcal K_N^{(a+b,a)}
	\left(
		z_\bk,z_\bq,\Delta;
		\mu_\bk,\mu_\bq-a
	\right),
	\qquad
	a, b\in\mathcal I
	}
\end{dmath*},
together with the three projected derivatives
\begin{dmath*}
	{
		\partial_\Delta,
		\qquad
		\partial_{z_\bq},
		\qquad
		\partial_\nu
	}
\end{dmath*}.
This is exactly parallel to the linear susceptibility calculation. In the linear case the angular contractions are generated by the one-variable incomplete-Bessel function \(G_\mu\). In the first nonlinear case they are generated by the projected bivariate kernels \(\mathcal K_N^{(r,a)}\), which encode the outer characteristic propagation and the full inner linear orbit response in a single object.

We now make the current-response formula explicit.  We use the circular basis
\begin{dmath*}
	\sigma,\rho
	\in {
	\left\{
		1,-1,0
	\right\}
	}
\end{dmath*}
and decompose the coefficients of the \(\bp\)-field as
\begin{dmath*}
	C_{\bp,b}^{(\alpha)}
	= {
	\sum_{\sigma\in\{1,-1,0\}}
	E_{\bp,\sigma}
	C_{\bp,b;\sigma}^{(\alpha)},
	\qquad
	\alpha\in\{\phi,\perp,\parallel\}
	}
\end{dmath*}.
The indices \(b\) and \(\sigma\) take values in the same set but have different meanings: \(b\) labels the gyrophase block of the nonlinear force operator, whereas \(\sigma\) labels the actual polarisation component of the field \(E_{\bp,\sigma}\).

Similarly, the inner amplitudes are written as
\begin{dmath*}
	X_{\bq,a}
	=
	\sum_{\rho\in\{1,-1,0\}}
	E_{\bq,\rho}
	X_{\bq,a;\rho}
\end{dmath*}.
The coefficients \(C_{\bp,b;\sigma}^{(\alpha)}\) are read off from the definitions of \(C_{\bp,c}^{(\alpha)}\), while \(X_{\bq,a;\rho}\) are read off from the linear response amplitudes.  In the following explicit formulas the species label is suppressed.

Explicitly,
\begin{dmath*}
	C_{\bp,b;1}^{(\phi)}
	=
	-\ii b
	\delta_{b,1}
	\frac{\ee^{\ii b\phi_\bp}}{\sqrt2}
	\frac{\mu_\bp\wc}{v_\perp}
	+
	\sum_{c=\pm1}
	\delta_{b+c,1}
	\frac{1-bc}{2}
	\left[
		-
		(1-b^2)
		\ii p_\perp c\sqrt2
		\ee^{\ii c\phi_\bp}
	\right]
\end{dmath*},
\begin{dmath*}
	C_{\bp,b;-1}^{(\phi)}
	=
	-\ii b
	\delta_{b,-1}
	\frac{\ee^{\ii b\phi_\bp}}{\sqrt2}
	\frac{\mu_\bp\wc}{v_\perp}
	+
	\sum_{c=\pm1}
	\delta_{b+c,-1}
	\frac{1-bc}{2}
	\left[
		-
		(1-b^2)
		\ii p_\perp c\sqrt2
		\ee^{\ii c\phi_\bp}
	\right]
\end{dmath*},
\begin{dmath*}
	C_{\bp,b;0}^{(\phi)}
	=
	\sum_{c=\pm1}
	\delta_{b+c,0}
	\frac{1-bc}{2}
	\left[
		-
		\ii c b^2
		\frac{p_\perp v_\parallel}{2v_\perp}
	\right]
\end{dmath*},
\begin{dmath*}
	C_{\bp,b;1}^{(\perp)}
	=
	b^2
	\delta_{b,1}
	\frac{\ee^{\ii b\phi_\bp}}{\sqrt2}
	\left(
		\Omega_\bp
		-
		p_\parallel v_\parallel
	\right)
\end{dmath*},
\begin{dmath*}
	C_{\bp,b;-1}^{(\perp)}
	=
	b^2
	\delta_{b,-1}
	\frac{\ee^{\ii b\phi_\bp}}{\sqrt2}
	\left(
		\Omega_\bp
		-
		p_\parallel v_\parallel
	\right)
\end{dmath*},
\begin{dmath*}
	C_{\bp,b;0}^{(\perp)}
	=
	\sum_{c=\pm1}
	\delta_{b+c,0}
	\frac{1-bc}{2}
	b^2
	\frac{p_\perp v_\parallel}{2}
\end{dmath*},
\begin{dmath*}
	C_{\bp,b;1}^{(\parallel)}
	=
	\delta_{b,1}
	b^2
	\frac{\ee^{\ii b\phi_\bp}}{\sqrt2}
	p_\parallel v_\perp
\end{dmath*},
\begin{dmath*}
	C_{\bp,b;-1}^{(\parallel)}
	=
	\delta_{b,-1}
	b^2
	\frac{\ee^{\ii b\phi_\bp}}{\sqrt2}
	p_\parallel v_\perp
\end{dmath*},
and
\begin{dmath*}
	C_{\bp,b;0}^{(\parallel)}
	=
	\delta_{b,0}
	\Omega_\bp
	-
	\sum_{c=\pm1}
	\delta_{b+c,0}
	\frac{1-bc}{2}
	b^2
	\frac{p_\perp v_\perp}{2}
\end{dmath*}.
The inner amplitudes are
\begin{dmath*}
	X_{\bq,1;1}
	= {
	\frac{\ee^{\ii\phi_\bq}}{\sqrt2}
	\mathcal G_{\bq,\Omega_\bq}f_0,
	\qquad
	X_{\bq,1;-1}
	=
	0,
	\qquad
	X_{\bq,1;0}
	=
	\frac{q_\perp}{2\Omega_\bq}
	\mathcal H f_0
	}
\end{dmath*},
\begin{dmath*}
	X_{\bq,-1;1}
	= {
	0,
	\qquad
	X_{\bq,-1;-1}
	=
	\frac{\ee^{-\ii\phi_\bq}}{\sqrt2}
	\mathcal G_{\bq,\Omega_\bq}f_0,
	\qquad
	X_{\bq,-1;0}
	=
	\frac{q_\perp}{2\Omega_\bq}
	\mathcal H f_0
	}
\end{dmath*},
and
\begin{dmath*}
	X_{\bq,0;1}
	= {
	0,
	\qquad
	X_{\bq,0;-1}
	=
	0,
	\qquad
	X_{\bq,0;0}
	=
	\partial_{v_\parallel}f_0
	}
\end{dmath*}.
For each projected current component, introduce the selector
\begin{dmath*}
	{
		S_1(v_\perp)
		\defeq
		\frac{v_\perp}{\sqrt2},
		\qquad
		S_{-1}(v_\perp)
		\defeq
		\frac{v_\perp}{\sqrt2},
		\qquad
		S_0(v_\parallel)
		\defeq
		v_\parallel
	}
\end{dmath*}.
The corresponding output harmonic indices are
\begin{dmath*}
	{
		N_1=-1,
		\qquad
		N_{-1}=1,
		\qquad
		N_0=0
	}
\end{dmath*}.
Then, after performing the gyrophase projection, the single-species current-response coefficient in the circular basis becomes
\begin{dmath*}
	\mathscr J_{\tau\sigma\rho}^{(2,s)}
	\left(
		\bk;\bp,\bq
	\right)
	=
	-
	\frac{e_s^3}{m_s^2\wc_s^2}
	\frac{1}{\Omega_\bp}
	\sum_{a\in\mathcal I}
	\sum_{b\in\mathcal I}
	\int_0^\infty
	\int_{-\infty}^{\infty}
	S_\tau
	\Bigg[
		\left(
			C_{\bp,b;\sigma}^{(\phi,s)}
			\mathbb K_{\bk,\bq,a,b;\tau}^{(\phi,s)}
			+
			C_{\bp,b;\sigma}^{(\perp,s)}
			\frac{q_\perp}{\wc_s}
			\mathbb K_{\bk,\bq,a,b;\tau}^{(\perp,s)}
			+
			C_{\bp,b;\sigma}^{(\parallel,s)}
			\frac{q_\parallel}{\wc_s}
			\mathbb K_{\bk,\bq,a,b;\tau}^{(\parallel,s)}
		\right)
		X_{\bq,a;\rho}^{(s)}
		+
		\mathbb K_{\bk,\bq,a,b;\tau}^{(s)}
		\left[
			C_{\bp,b;\sigma}^{(\perp,s)}
			\partial_{v_\perp}
			+
			C_{\bp,b;\sigma}^{(\parallel,s)}
			\partial_{v_\parallel}
		\right]
		X_{\bq,a;\rho}^{(s)}
	\Bigg]
	v_\perp
	\dd v_\perp
	\dd v_\parallel
\end{dmath*},
where $\bk = \bp + \bq$ is understood, $s$ identifies the species and
\begin{dmath*}
	\mathbb K_{\bk,\bq,a,b;\tau}^{(s)}
	\defeq
	\mathcal K_{N_\tau+b}^{(a+b,a)}
	\left(
		z_{\bk,s},z_{\bq,s},\Delta;
		\mu_{\bk,s},\mu_{\bq,s}-a
	\right)
\end{dmath*},
\begin{dmath*}
	\mathbb K_{\bk,\bq,a,b;\tau}^{(\phi,s)}
	\defeq
	-
	\partial_\Delta
	\mathcal K_{N_\tau+b}^{(a+b,a)}
	\left(
		z_{\bk,s},z_{\bq,s},\Delta;
		\mu_{\bk,s},\mu_{\bq,s}-a
	\right)
\end{dmath*},
\begin{dmath*}
	\mathbb K_{\bk,\bq,a,b;\tau}^{(\perp,s)}
	\defeq
	\partial_{z_{\bq,s}}
	\mathcal K_{N_\tau+b}^{(a+b,a)}
	\left(
		z_{\bk,s},z_{\bq,s},\Delta;
		\mu_{\bk,s},\mu_{\bq,s}-a
	\right)
\end{dmath*},
\begin{dmath*}
	\mathbb K_{\bk,\bq,a,b;\tau}^{(\parallel,s)}
	\defeq
	\left.
	\partial_\nu
	\mathcal K_{N_\tau+b}^{(a+b,a)}
	\left(
		z_{\bk,s},z_{\bq,s},\Delta;
		\mu_{\bk,s},\nu
	\right)
	\right|_{\nu=\mu_{\bq,s}-a}
\end{dmath*}.
The total nonlinear current is the sum of the single-species current coefficients.  Hence, comparing with the defining relation for the nonlinear susceptibility, one obtains
\begin{dmath*}
	\chi_{\tau\sigma\rho}^{(2)}
	\left(
		\bp,\Omega_\bp;
		\bq,\Omega_\bq
	\right)
	=
	\frac{\ii}{\epsilon_0\Omega_\bk}
	\sum_s
	\mathscr J_{\tau\sigma\rho}^{(2,s)}
	\left(
		\bk;\bp,\bq
	\right)
\end{dmath*}.
Here \(\tau\in\{1,-1,0\}\) labels the output current component.  This formula gives the nonlinear susceptibility explicitly as a finite sum over the nine elementary bivariate blocks \((a,b)\), followed by a two-dimensional velocity integral. The Cartesian tensor is obtained from \(\chi_{\tau\sigma\rho}^{(2)}\) by the fixed change of basis between the circular basis \(\{\mathbf e_1,\mathbf e_{-1},\mathbf e_0\}\) and the Cartesian basis.

\section{Outlook toward higher orders}\label{sec:higher-orders}

The factorisation
\begin{dmath*}
	f_\bk^{(2)}
	=
	\ee^{\ii z_\bk\sin\psi_\bk}
	\tilde f_\bk^{(2)}
\end{dmath*}
has the same form as the linear factorisation. This is not accidental. In the weakly nonlinear hierarchy considered here, the left-hand characteristic operator is the same at every perturbative order; only the source term changes. Consequently, once the source at a given order has been written in terms of orbit kernels, the next order is obtained by applying the same characteristic integral once more.

The linear response is governed by the one-variable incomplete-Bessel kernel
\begin{dmath*}
	G_{\mu_\bk}
	\left(
		z_\bk,
		\psi_\bk
	\right)
\end{dmath*}.
The first nonlinear response is governed by the bivariate kernel
\begin{dmath*}
	G_{\mu_\bk,\mu_\bq-a}^{(a+b)}
	\left(
		z_\bk,\psi_\bk;
		z_\bq,\psi_\bq
	\right)
\end{dmath*},
which couples the outer characteristic of the output mode \(\bk\) to the inner linear orbit associated with the mode \(\bq\). At the next perturbative step, the nonlinear source contains the first nonlinear response itself. Applying the outer characteristic integral to this source therefore produces a trivariate kernel, schematically of the form
\begin{dmath*}
	G_{\mu_\bk,\mu_{\bq_1},\mu_{\bq_2}}^{(\mathbf r)}
	\left(
		z_\bk,\psi_\bk;
		z_{\bq_1},\psi_{\bq_1};
		z_{\bq_2},\psi_{\bq_2}
	\right)
\end{dmath*},
with discrete shifts \(\mathbf r\) determined by the circular and parallel polarisation blocks appearing in the nested nonlinear source.

More generally, within the fixed-characteristic perturbative hierarchy, the \(n\)-th order distribution is expected to be expressible in terms of \(n\)-variate orbit kernels. These kernels are obtained recursively by adding one characteristic time integral for each nonlinear interaction. Their harmonic expansions reproduce the usual multiple cyclotron-harmonic sums, while their projected versions provide the angular contractions needed for the corresponding higher-order susceptibility tensors.

This recursive statement has the standard limitation of weakly nonlinear response theory. It assumes that the nonlinear corrections remain on the right-hand side of the Vlasov equation, so that the unperturbed cyclotron characteristics are not modified. If nonlinear orbit corrections, trapping, ponderomotive drifts, or self-consistent changes of the guiding-centre motion are promoted to the left-hand side, then the characteristic operator itself changes and the simple incomplete-Bessel recursion is deformed.

Thus the present construction should be understood as a compact orbit-kernel reformulation of the fixed-characteristic perturbative hierarchy. In this setting the first nonlinear susceptibility is generated by projected bivariate kernels, and higher-order susceptibilities should be generated by projected multivariate kernels. Whether these higher-rank projected kernels admit closed forms as compact as those obtained here depends on the corresponding angular-contraction calculus.

\section{Conclusion}\label{sec:conclusion}

We have reorganised the first nonlinear Vlasov--Maxwell response around the characteristic structure of the problem.  The linear response is obtained by applying one cyclotron orbit integral to one Larmor phase, producing the incomplete-Bessel function \(G_\mu\).  The first nonlinear response is obtained by applying an outer cyclotron orbit integral to an inner \(G_{\mu_\bq}\)-resolved linear response.  This produces the bivariate incomplete-Bessel function
\begin{dmath*}
	G_{\mu,\nu}^{(r)}
	\left(
		z,\psi;
		w,\chi
	\right)
\end{dmath*}.
The bivariate function is therefore not an auxiliary notation, but the natural orbit-resolvent kernel associated with the first nonlinear cyclotron response.

The resulting nonlinear distribution function is explicit and preserves all physical factors of the original calculation through compact differential blocks.  Expanding the bivariate function recovers a generalisation of the traditional double-harmonic Liu--Tripathi expression, while keeping it unexpanded gives a direct route to the nonlinear susceptibility tensor through projected bivariate angular contractions.  In this sense the bivariate formulation is equivalent to the harmonic representation, but reorganises it around the geometry of the characteristic flow.

The projected kernels introduced above are also relevant for numerical implementation.  In the traditional formulation one has to handle double cyclotron-harmonic sums, together with several shifted Bessel factors and derivative-generated recurrence terms.  In the present formulation these structures are absorbed into a small family of projected kernels,
\begin{dmath*}
	{
		\mathcal K_N^{(r,a)},
		\qquad
		\partial_\Delta\mathcal K_N^{(r,a)},
		\qquad
		\partial_{z_\bq}\mathcal K_N^{(r,a)},
		\qquad
		\partial_\nu\mathcal K_N^{(r,a)}
	}
\end{dmath*}.
These kernels may be evaluated either from their angular-integral definition, from their single-sum coefficient representation in terms of the contracted coefficients \(T_\ell\), or from the original characteristic integral when this is more stable.  This flexibility is expected to be advantageous in regimes where direct double-harmonic summation is slowly convergent, strongly oscillatory, or numerically delicate because of cancellations among shifted harmonics.

The numerical advantage should not be interpreted as a universal reduction of computational cost.  Its usefulness depends on the values of the Larmor radii, detunings, damping prescription, and truncation strategy.  Nevertheless, the bivariate formulation offers a more structured computational target: instead of repeatedly summing and differentiating long harmonic expressions, one can tabulate projected kernels and their derivatives, reuse them across tensor components, and exploit their recurrence, resolvent, and projection identities.  This suggests a practical route toward efficient evaluation of nonlinear susceptibilities in parameter regimes where the harmonic representation is algebraically unwieldy.

The main advantage of the formulation is therefore both structural and potentially computational.  Structurally, it identifies the special-function object underlying the first nonlinear cyclotron response.  Computationally, it replaces the unorganised double-harmonic expansion by a small set of orbit-projected kernels whose analytic properties can be studied independently and then used systematically in the susceptibility tensor.

\section{Declaration of generative AI and AI-assisted technologies in the manuscript preparation process}

During the preparation of this work the author used ChatGPT in order to double-check the correctness of part of the calculations. After using this tool, the author reviewed and edited the content as needed and takes full responsibility for the content of the published article.

\bibliographystyle{unsrtnat}
\bibliography{nonlinear_vlasov_bivariate_incomplete_bessel}

\end{document}